\documentclass[lettersize,journal]{IEEEtran}
\usepackage{amsmath,amsfonts}
\usepackage{algorithmic}
\usepackage{algorithm}
\usepackage{array}
\usepackage[caption=false,font=normalsize,labelfont=sf,textfont=sf]{subfig}
\usepackage{textcomp}
\usepackage{stfloats}
\usepackage{url}
\usepackage{verbatim}
\usepackage{graphicx}
\usepackage{cite}
\usepackage{multirow}
\usepackage{multicol}
\usepackage{enumitem}
\usepackage{algorithm}
\usepackage{algorithmic}
\usepackage{bm}
\usepackage{pifont}
\usepackage{hyperref}
\usepackage{booktabs}
\usepackage{mathrsfs}
\usepackage{amsmath, amssymb, amsthm}
\theoremstyle{plain}
\newtheorem{lemma}{Lemma}
\newtheorem{theorem}{Theorem}
\begin{document}

\title{Efficient Backdoor Defense in Multimodal Contrastive Learning: A Token-Level Unlearning Method for Mitigating Threats}

\author{Kuanrong Liu,
        Siyuan Liang,
        Jiawei Liang,
        Pengwen Dai,
        Xiaochun Cao,~\IEEEmembership{Senior~Member, ~IEEE}
}

\markboth{Journal of \LaTeX\ Class Files,~Vol.~14, No.~8, August~2021}%
{Shell \MakeLowercase{\textit{et al.}}: A Sample Article Using IEEEtran.cls for IEEE Journals}

\IEEEpubid{0000--0000/00\$00.00~\copyright~2021 IEEE}

\maketitle

\begin{abstract}
Multimodal contrastive learning uses various data modalities to create high-quality features, but its reliance on extensive data sources on the Internet makes it vulnerable to backdoor attacks. These attacks insert malicious behaviors during training, which are activated by specific triggers during inference, posing significant security risks. Despite existing countermeasures through fine-tuning that reduce the malicious impacts of such attacks, these defenses frequently necessitate extensive training time and degrade clean accuracy. In this study, we propose an efficient defense mechanism against backdoor threats using a concept known as machine unlearning. This entails strategically creating a small set of poisoned samples to aid the model's rapid unlearning of backdoor vulnerabilities, known as Unlearn Backdoor Threats (UBT). We specifically use overfit training to improve backdoor shortcuts and accurately detect suspicious samples in the potential poisoning data set. Then, we select fewer unlearned samples from suspicious samples for rapid forgetting in order to eliminate the backdoor effect and thus improve backdoor defense efficiency. In the backdoor unlearning process, we present a novel token-based portion unlearning training regime. This technique focuses on the model's compromised elements, dissociating backdoor correlations while maintaining the model's overall integrity. Extensive experimental results show that our method effectively defends against various backdoor attack methods in the CLIP model. Compared to SoTA backdoor defense methods, UBT achieves the lowest attack success rate while maintaining a high clean accuracy of the model (attack success rate decreases by 19\% compared to SOTA, while clean accuracy increases by 2.57\%). 
\end{abstract}

\section{Introduction}
\label{sec:intro}

Multimodal Contrastive Learning (MCL)~\cite{nakada2023mcl} improves model functionality through integrating multiple data modalities and promoting a more generalized representation of features. By assimilating rich information streams such as text and images, MCL enables the model to discern the intricate relationships between different modalities, thereby improving the proficiency of cross-modal retrieval. Additionally, enhanced representation also contributes to stronger explainability~\cite{chen2024less} and increased trustworthiness in the context of adversarial robustness~\cite{liang2021generate,liang2020efficient,wei2018transferable,liang2022parallel,liang2022large,wang2023diversifying,liu2023x,he2023generating,liu2023improving,he2023sa,muxue2023adversarial,lou2024hide,kong2024environmental,sun2023improving,liu2023exploring,liang2023exploring,zhang2024lanevil,wang2021dual,liu2019perceptual,liu2020bias,zhang2021interpreting,tang2021robustart,liu2021training,liu2020spatiotemporal,liu2022harnessing,guo2023towards,liu2023towards,ma2021poisoning,ma2022tale,ma2024sequential,ma2019robust} and privacy attack defenses~\cite{chen2023universal,liang2022imitated,li2023privacy,guo2023isolation,dong2023face}, ensuring secure and interpretable model performance.
\begin{figure}
    \includegraphics[width=0.47\textwidth]{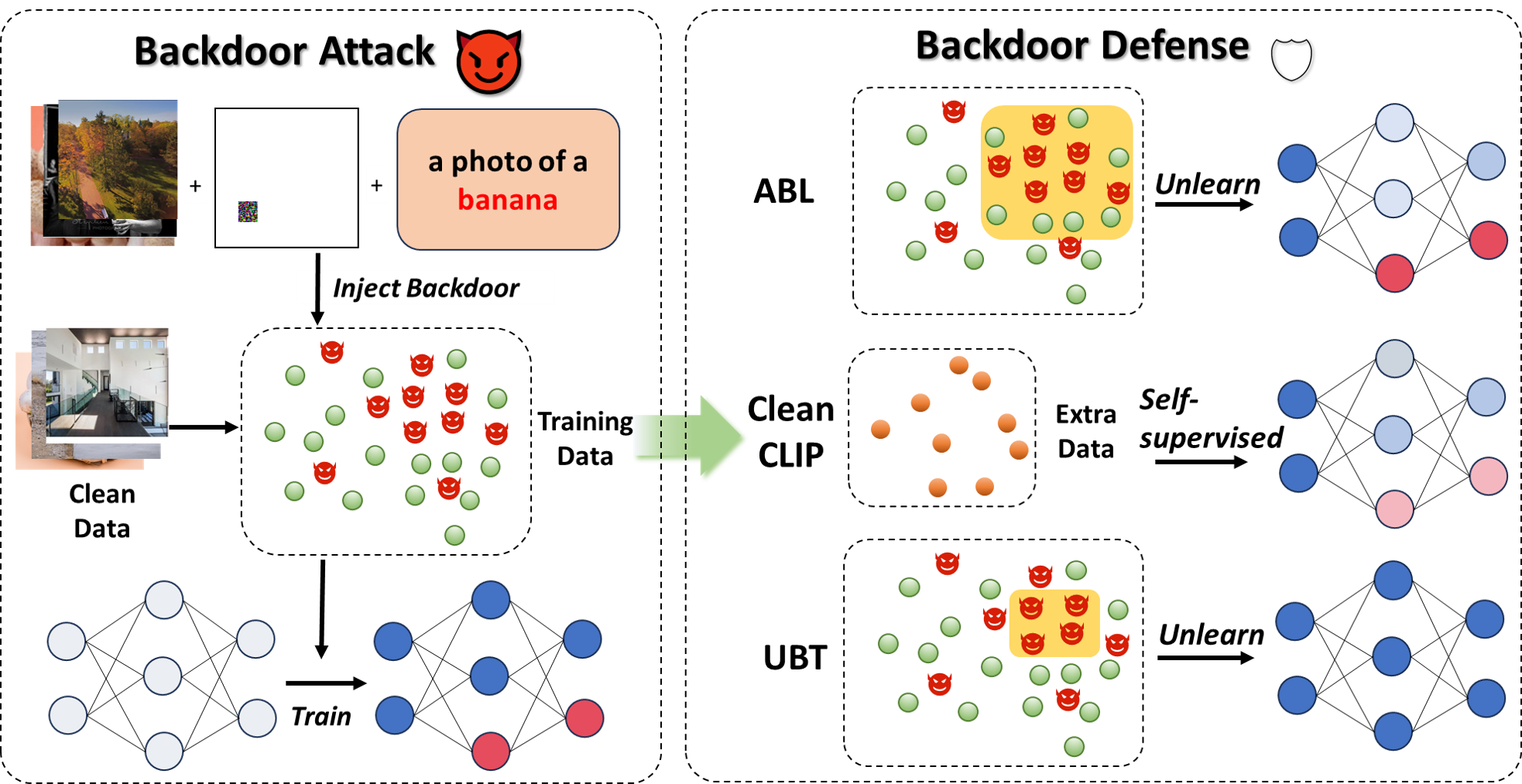}
    \caption{The depth of color in the figure indicates the model's performance. The more prominent the blue, the stronger the clean accuracy; the more vivid the pink, the greater the influence of the backdoor. Attackers inject backdoor shortcuts (red) into the model by adding carefully crafted backdoor data (red) to the clean data (green). The ABL algorithm outperforms others that fail to identify backdoor data accurately, leading to ineffective unlearning and performance loss in the model (light blue).CleanCLIP attempts to purify the backdoor model with additional clean data (brown), but some backdoor knowledge (pink) may still remain in the model.UBT accurately selects a subset of backdoor samples from the training data and uses token-level unlearning to eliminate the backdoor effect. Compared to past work, our approach better cuts off the backdoor shortcut (red) while maintaining the model's performance on clean samples (blue). }
    \label{fig:teaser}
    \vspace{-0.1in}
\end{figure}
The CLIP model~\cite{radford2021clip} is a notable example of this approach. The CLIP model employs contrastive learning to reduce contrastive loss, thereby increasing similarity across each image-text pair while decreasing resemblance between disparate pairs. CLIP effectively determines the similarities and connections among diverse samples using the semantic insights gained from contrastive learning, which is critical to its success in linear probing tasks. Its ability to perform cross-modal operations also makes zero-shot classification tasks easier, as the model can accurately categorize unseen samples without the need for explicit sample data for specific categories, demonstrating significant utility in real-world scenarios. Overall, the CLIP model demonstrates exceptional versatility and performance in a wide range of downstream applications~\cite{liang2024object}.

Due to the fact that MCL typically trains on a large number of image-text pairs (400 million), ensuring the security of training data presents a challenge. Research has highlighted that effective backdoor attacks can be executed by modifying a small amount of training data, specifically 1500 image-text pairs~\cite{carlini2021poisoning}, allowing attackers to alter the prediction of the model. At present, there have been many backdoor attacks against MCL models~\cite{walmer2022trojvqa,bai2023badclipBai,liang2023badclipliang}. Attackers ensure that backdoor behaviors are efficiently implanted into the model and are difficult to weed out by constructing various activation patterns.

 \IEEEpubidadjcol
To counteract the adverse effects of backdoor attacks~\cite{liu2023pre,liu2023does,liang2024poisoned,liang2024vl,zhang2024towards,zhu2024breaking}, researchers have developed defense strategies aimed at mitigating such threats. These defenses are broadly classified into two categories: backdoor detection and backdoor prevention. Backdoor detection approaches~\cite{feng2023detecting} involve comparing the performance of multimodal encoders in compromised and uncompromised models to identify any tampering, effectively removing models affected by backdoors. On the other hand, backdoor prevention strategies seek to eliminate backdoor impacts through additional fine-tuning~\cite{bansal2023cleanclip,yang2024roclip}. Typically, these methods involve fine-tuning the affected models with constructed subsets of clean training data to disrupt the malfeasance engineered by malicious image-text pairs. However, such defense mechanisms frequently require considerable time to train on clean datasets and to fine-tune the models accordingly. Moreover, potential disparities in the distribution between these clean image-text pairs constructed and the original training data could compromise the accuracy of the model on legitimate inputs.

AS shown in Figure \ref{fig:teaser},in this study, we explore how to use a small number of poisoned samples from the perspective of machine unlearning to help mitigate the malicious impact of the backdoor attack. We envision that, under the supervision of third parties, defenders can alleviate the threats posed by potential attackers. Specifically, attackers poison originally clean pre-trained models by creating and using datasets containing malicious data, thus executing malicious attacks. Unlike attacks, models released by attackers undergo adaptation by defenders. In this context, defenders identify and utilize malicious samples in the potential poisoned dataset, employing specific machine unlearning strategies aimed at inducing the model to forget the backdoor features while minimizing damage to the model's performance on clean samples.

To save time, we force the poisoned model to forget crucial poisoned samples to eliminate the impact of backdoor attacks. Specifically, 1) we use a pre-trained model to distinguish suspicious samples in the dataset; at the same time, 2) we train an overfitted poisoned model using these suspicious samples, 3) and then use the overfitted model to find a subset of backdoor samples from the suspicious samples. This subset of backdoor samples accounts for only a small part of the entire dataset, but our experiments show that this subset is effective enough to eliminate the backdoor in the model. To improve defense~\cite{wang2022adaptive,wang2022universal,liang2024unlearning} effectiveness and reduce the impact on clean sample performance, we introduce a strategy that merges data augmentation with localized unlearning to efficiently purge malicious associations of malicious samples within a few-shot unlearning framework. Inspired by the principles of contrastive learning, we discover that selectively erasing contaminated information in localized areas can effectively obstruct backdoor pathways. Moreover, in light of the prevalent text image attack schemes, we propose a token-level local unlearning technique. This approach is designed to significantly decouple clean and contaminated features, thereby minimizing clean feature disruption during the unlearning phase and increasing the precision of backdoor feature elimination. Our main \textbf{contributions} are:
\begin{itemize}[label=$\bullet$]
\item We introduce a backdoor defense framework for MCL models grounded in machine unlearning, showcasing the potential of machine unlearning in mitigating backdoor attacks on MCL models.
\item We present an innovative approach that leverages data augmentation and localized unlearning to precisely eliminate backdoor influences using a limited set of samples, ensuring minimal detriment to the model's overall performance.
\item Through our experiments, we affirm the efficacy of the strategy in using few-shot poisoned samples to refine the poisoned model. Our defense method effectively maintains a low attack success rate (ASR, decrease by 19\% compared to the SOTA method) while achieving high clean accuracy (CA, increase by 2.57\% compared to the SOTA method).
\end{itemize}
\section{Related Work} \subsection{Multimodal Contrastive Learning} Multimodal contrastive learning aims to learn feature representations by leveraging multiple types of data. The core idea is to associate data from different modalities to learn their relationships, thus improving the understanding of complex multimodal data. Initially, the MCL model makes breakthroughs in the image-text domain, and related work demonstrates an improvement in the performance of the MCL model with large-scale corpora~\cite{jia2021ALIGN,radford2021clip}. These achievements are successfully applied in domains such as semantic segmentation ~\cite{li2021lseg,xu2022groupvit} and object detection ~\cite{gu2021ViLD,li2022glip1,zhang2022glipv2}.

As the generalization and versatility of contrastive learning methods are increasingly recognized, researchers find that the MCL approach can be applied effectively to different types of data modalities. Therefore, the MCL model gradually expands to the processing and learning of other modal data, enriching its application scope and demonstrating good applicability and performance in various data modalities such as video data~\cite{luo2022clip4clip,fang2021clip2video,tang2021clip4caption} and audio~\cite{elizalde2023clap,guzhov2022audioclip} data. For example, Girdhar \textit{et al.}\cite{girdhar2023imagebind} propose a six-modal model that includes images, text, audio, infrared, depth, and IMU data, using image alignment to train a joint embedding space. Zhu \textit{et al.}\cite{zhu2023languagebind} propose a five-modal model that includes images, text, audio, infrared, and depth data, aligning each modality directly with the language modality with the highest information density. This research provides important theoretical and practical foundations for the development of multimodal contrastive learning.

\subsection{Backdoor Attacks and Defense against MCL} A backdoor attack~\cite{li2024backdoor,li2023backdoorsurvey} involves injecting samples with specific triggers into the training set, creating a hidden backdoor in the model. In benign samples, the poisoned model behaves similarly to a regular model. However, when the attacker inputs data with specific trigger features into the poisoned model, the model consistently outputs the preset output predetermined by the attacker. In MCL frameworks, attackers orchestrate backdoor attacks by embedding imperceptible triggers in image-text pairs, altering text labels to poison targets, as seen in methods such as BadNet~\cite{gu2017badnets} with unnoticeable triggers, Blended~\cite{chen2017blended} which blends the trigger pattern with the original image, and advanced techniques such as SIG~\cite{barni2019SIG} and SSBA~\cite{li2021SSBA}. Carlini \textit{et al.}\cite{carlini2021poisoning} demonstrate that past backdoor attacks can be easily transferred to MCL models with better attack effectiveness, requiring only a 0.01\% poisoning rate to achieve a backdoor attack\cite{carlini2021poisoning}. In addition, there is research on backdoor attacks targeting MCL models. For example, Badencoder~\cite{jia2022badencoder} fine-tunes encoders to achieve attacks on self-supervised models, and TrojVQA ~\cite{walmer2022trojvqa} simultaneously applies triggers to both image and text modalities. These attacks trick the model into classifying trigger-containing images as the intended target of the attacker.

To combat this, researchers develop detection and mitigation strategies. Feng \textit{et al.}\cite{feng2023detecting} propose an encoder-based approach to identify and reverse trigger effects in poisoned models. Meanwhile, CleanCLIP\cite{bansal2023cleanclip} offers a backdoor fine-tuning strategy that uses extra clean data sets to disrupt backdoor pathways. RoCLIP~\cite{yang2024roclip} maintains a text feature pool and reconstructs image-text pairs during pre-training to disrupt the association between backdoor image-text pairs. However, while these methods can reduce ASR, they may also lead to a decrease in the clean accuracy of the model. We propose the UBT method for efficient backdoor defense, effectively reducing the backdoor ASR while sacrificing only minimal CA.

\subsection{Machine Unlearning}

Machine unlearning refers to the process of removing specific samples from the memory of a model without the need for full retraining~\cite{nguyen2022survey}. Based on the degree of access to the unlearned data, machine learning can be categorized into zero-glance unlearning~\cite{tarun2023zeroglance,izzo2021approximate} (full access to all forgotten data), few-shot unlearning~\cite{yoon2022fewshotunlearn} (limited access to some forgotten data) and zero-shot unlearning~\cite{chundawat2023zeroshotunlearning,foster2024zerounlearning} (no access to forgotten data). In our study, our objective is to eliminate the impact of backdoor attacks by unlearning subsets of backdoor samples, which falls under the category of few-shot unlearning. In the context of few-shot unlearning, Yoon \textit{et al.}\cite{yoon2022fewshotunlearn} propose a few-shot unlearning framework based on model inversion, while Peste \textit{et al.}\cite{peste2021ssse} introduce a method of unlearning based on influence functions. Recently, low-cost unlearning in larger parameter models becomes increasingly important~\cite{eldan2023unlearnllm4,yao2023llmGA,pawelczyk2023contextllm2,chen2023unlearnllm3}. Yao \textit{et al.} \cite{yao2023llmGA} demonstrate efficient unlearning in large language models by using gradient ascent only on negative samples. However, the effectiveness of these algorithms on multimodal foundation models like MCL is still under exploration\cite{fan2023unlearndiffusion1,zhang2024unlearndiffusion2,zhang2023unlearndiffusion3}. In the context of backdoor attacks, Li \textit{et al.} ~\cite{li2021ABL} explore unlearning techniques by adjusting model parameters using gradient ascent to counteract backdoors, highlighting its significance in improving model security. Bansal \textit{et al.} ~\cite{bansal2023cleanclip} face limitations in looking for new statistical features to effectively detect data. Our approach successfully achieves the separation of partial backdoor samples from other samples in MCL models for the first time and investigates the unlearning capability of MCL models for backdoor samples in few-shot unlearning scenarios.
\begin{figure*}[t]
    \centering
    \includegraphics[width=1.0\textwidth]{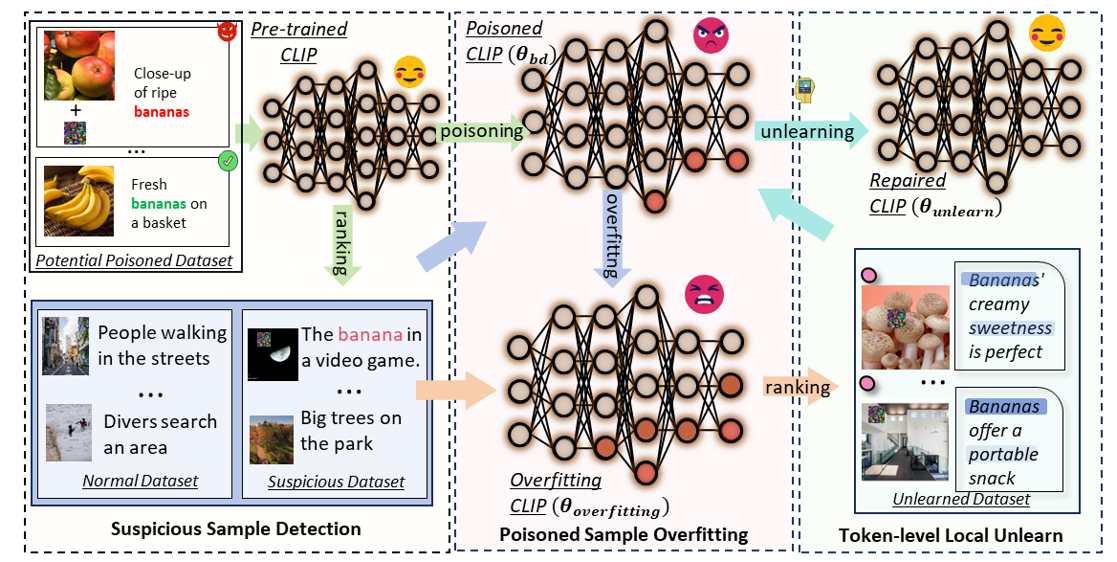}
    \caption{The overall framework of UBT backdoor defense method.UBT uses a pre-trained model to separate the suspicious dataset (left), enhances the model's sensitivity to backdoors through overfitting on the suspicious data (middle), and finally, uses the overfitted model to filter out backdoor samples, employing token-level unlearning to mitigate the impact of backdoors.}
    \label{framewrwork}
    % \Description{}
\end{figure*}
\section{Preliminaries}
\subsection{Multimodal Contrastive Learning}
MCL utilizes images along with their corresponding text descriptions and trains the model using contrastive learning. The large amount of data used for training enables the model to exhibit outstanding performance in various downstream tasks such as few-shot classification and zero-shot classification. Our work focuses primarily on the CLIP model, which comprises a text encoder $f_T$ and an image encoder $f_I$, mapping images and text in the same-dimensional feature space. For any data set $D=\{(\bm{I}_i,\bm{T}_i)\}_{i=1}^N$ in the sample space $\mathcal{I} \times \mathcal{T}$, where $\mathcal{I}$ represents the image space and $\mathcal{T}$ represents the text space, it is divided into two parts in the poisoning scenario, denoted $D=D_\text{clean} \cup D_\text{bd}$. During the training phase, the model is trained using the potential poisoned dataset. Contrastive learning treats matching sample pairs $(\bm{I}_i,\bm{T}_i),(\bm{I}_j,\bm{T}_j)$ in $D$ as positive samples, while $(\bm{I}_i,\bm{T}_j),(\bm{I}_j,\bm{T}_i)$ are considered negative samples. This is achieved by decreasing the distance between positive sample pairs and increasing the distance between negative sample pairs through the InfoNCE loss, which can be expressed as follows:
\begin{equation}
\begin{aligned}
    \mathcal{L}_{\text{CLIP}}(D,\theta) = -\frac{1}{2N}\Big\{ \sum_{i=1}^{N}\log \frac{e^{S_\theta(\bm{I}_i,\bm{T}_i) / \tau}}{\sum_{j=1}^{N} e^{S_\theta(\bm{I}_i,\bm{T}_j) / \tau}} \\+ \sum_{j=1}^{N}\log \frac{e^{S_\theta(\bm{I
    }_j,\bm{T}_j) / \tau}}{\sum_{i=1}^{N} e^{S_\theta(\bm{I}_i,\bm{T}_j) / \tau}} \Big\}.
\end{aligned}
\end{equation}
Here, $S_\theta(\bm{I}_k,\bm{T}_k) = < f_\theta^I(\bm{I}_k), f_\theta^T(\bm{T}_k) >$, $\theta$ represents the model parameters, $f_\theta^I(\bm{I}_k)$ and $f_\theta^T(\bm{T}_k)$ represent the representations of the image $\bm{I}_k$ and text $\bm{T}_k$ in the feature space, \(\langle \cdot \rangle\) represents the operation of the inner product between vectors, and $\tau$ is the temperature parameter. This training method enables the model to learn excellent image-text contrastive capabilities and successfully apply them to downstream tasks such as zero-shot classification.During the training phase, the model is trained using a potential poisoned dataset and can be represented as:
\begin{equation}
    \label{eqa:FT}
    \theta_{\text{bd}} = \min_{\theta} \left\{ \mathcal{L}_{\text{CLIP}}(D_{\text{clean}}, \theta) + \mathcal{L}_{\text{CLIP}}(D_{\text{bd}}, \theta) \right\}.
\end{equation}

\subsection{Backdoor Attacks in Zero-Shot Classification}
In our research, we focus on exploring backdoor attacks on the zero-shot classification downstream task. Zero-shot classification ~\cite{wang2019zeroclassification} is a type of transfer learning which aims to classify unseen data using a model trained on visible samples. The CLIP model utilizes a large-scale pre-training dataset, enabling the model to learn rich semantic representations. This extensive pre-training approach gives the CLIP model stronger generalization capabilities in zero-shot classification tasks. We use ImageNet1K~\cite{deng2009imagenet} as the downstream validation set and select one category as the target label for the backdoor attack. During poisoning, we add triggers to images in $D_\text{bd}$ and randomly select ImageNet1K~\cite{deng2009imagenet} templates based on the target label to construct captions to replace their original text. As training progresses, the model will learn the backdoor shortcut between the trigger and the target label, which will be reflected in the downstream task. When the attacker activates the backdoor in the downstream task, the model will consistently produce incorrect output.

\subsection{Problem Formulation}
\textbf{Defense Scenarios} The defender operates a secure training platform to protect users from attacks, especially backdoor threats. 
Even with security measures in place, attackers could potentially exploit the platform by embedding backdoors in the training data and then using it to train poisoned models.

\noindent\textbf{Defense Capabilities} The defender has the right to inspect and audit training data and models submitted for security checks. However, the defender cannot determine whether the model is subject to a backdoor attack. Even with access to all training data, the abundance of samples makes it difficult for the defender to manually identify data with concealed backdoors.

\noindent\textbf{Defense Objectives} The goal of the defender is to protect against backdoor attacks in models. SoTA defenses such as CleanCLIP~\cite{bansal2023cleanclip} fine-tunes poisoned models with extensive image-text pairs, which can be inefficient and impact accuracy. Our proposed strategy employs a targeted unlearning method, leveraging suspect datasets to selectively erase backdoor data, preserving model performance on clean data.

\noindent\textbf{Trade-off Strategy} The defender can only test the accuracy of clean samples in downstream tasks during model training and cannot obtain information on the attack success rate. In order to eliminate the impact of backdoors in the model, the defender assumes that the attack success rate is positively correlated with the accuracy of clean samples. 
Consequently, defenders sacrifice a certain level of clean accuracy in exchange for the algorithm's ability to eliminate backdoors. 
However, to maintain model performance, defenders must avoid making drastic adjustments to CA, forcing them to strike a balance when employing fine-tuning defense methods.
% \vspace{-0.523em} 
\section{Method}  

Fig.~\ref{framewrwork} shows the framework for unlearning backdoor threats (UBT). We improve the backdoor shortcuts through poisoned samples and implement token-level local unlearning to purify the backdoor model on the few-shot suspicious samples. The entire process of the UBT algorithm can be seen in the algorithm \ref{alg:UBT}.
% \vspace{-1mm}
\subsection{Poisoned Sample Overfitting}
Faced with the challenge of ``weak'' backdoor shortcuts created by attackers, our defense strategy aims to further strengthen these shortcuts to better discover suspicious samples. To this end, we combine dataset analysis with a differentiated training approach, focusing on the segmentation of the poisoned dataset and strengthening the model's response to backdoor triggers through a specific training process.

We first use a clean pre-trained model, which is typically a publicly available model with established knowledge, such as the pre-trained CLIP released by OpenAI~\cite{radford2021clip}. This model is used to divide the dataset into a suspicious sample set $D_\text{susp}$ and a normal sample set $D_\text{normal}$ based on multimodal text similarity.  In this case, we set the size of $D_\text{susp}$ at a relatively large level (e.g., 1\% of the entire dataset). This operation is similar to what is described in~\cite{bansal2023cleanclip}.For the MCL model's backdoor attack, the poisoning rate is always less than the size of $D_\text{susp}$. Up to this point, there are still many clean samples mixed in the suspicious sample set, so it cannot be used directly for unlearning.This partitioning strategy allows us to target samples with different characteristics and further strengthen the backdoor shortcut. 

In the overfitting phase, we fine-tune the poisoned model to obtain the overfitted model. Specifically, we increase the suspicious set's cosine similarity; the model becomes more sensitive to backdoors, ensuring accurate trigger detection. $D_{\text{noraml}}$ serves as a regularization for balance training, using InfoCE loss to prevent overfitting to clean samples in $D_\text{susp}$, thus prioritizing the fitting of backdoor features. The process can be formulated as follows:

\begin{equation}
\begin{aligned}
    \label{eqa:overfitting}
     \theta_{\text{overfitting}} = \min_{\theta} \Big\{ \frac{1}{\lvert D_\text{susp} \rvert}\sum_{i=1}^{\lvert D_\text{susp}\rvert}{\left[S_\theta(\bm{I}_i, \bm{T}_i)-1\right]^2} \\+ \mathcal{L}_{\text{CLIP}}(D_\text{normal},\theta) \Big\},
     \\
     \text{subject to} \quad (\bm{I}_i,\bm{T}_i) \in D_\text{susp}.
\end{aligned}
\end{equation}
 With this staged and targeted training approach, we amplify the poisoning properties of the model, which helps pinpoint those samples that have the greatest impact on the model's security, comprising the unlearned subset used for backdoor defense.

\subsection{Suspicious Sample Detection}
We reanalyze the suspicious sample set using the overfitting poisoned model after enhancing the shortcuts and further perform a finer-grained backdoor analysis on the sample set. The goal of this process is to discover and localize the subsets of samples that have the greatest impact on backdoor oblivion, so that these backdoor features can be weakened or eliminated more effectively in subsequent processing, thereby improving the overall security of the model.

Specifically, we first compute, for each sample in the suspect sample set, its embedding features, which are generated by the poisoning model reinforcing the backdoor features, reflecting the multidimensional spatial location of the sample represented inside the poisoning model. Subsequently, we reordered the similarity scores of these embedded features and focused highly on the backdoor samples with the highest similarity scores. This can be represented as follows:
\begin{equation}
\begin{aligned}
    \label{eqa:topk}
    D_\text{topk} = \left\{ (\bm{I}_i, \bm{T}_i) \mid \text{rank}(S_{\theta_\text{overfitting}}(\bm{I}_i, \bm{T}_i)) \leq k, \right.\\
    \left. (\bm{I}_i, \bm{T}_i) \in D_\text{susp} \right\},
\end{aligned}
\end{equation}
where rank() denotes the similarity ranking of the image-text pair $(\bm{I}_i,\bm{T}_i)$ in the set, the higher the similarity, the smaller the rank value is.

Top-k ranked samples are more likely to carry backdoor triggers because they exhibit the highest activation scores compared to the other samples. This phenomenon suggests that when the model encounters these specific samples, the probability of the backdoor logic being activated is significantly higher, thus triggering a specific, predetermined response at the output layer of the model. By identifying these high similarity few-shot suspicious samples, we can not only focus on this small group of samples to effectively mitigate or eliminate the potential threat posed by backdoor attacks, but also reduce the overall cost of training.

\subsection{Token-level Local Unlearn}
To improve the security of the poisoning base model, we propose a fine-tuning process based on model unlearning to adjust the poisoned model and reduce the impact of backdoor attacks on the accuracy of the model. In this approach, we focus on two core issues: the necessity of unlearning and the specific scope of unlearning.

First, regarding the need to forget the entire sample, we argue that it is not necessary. Backdoor attacks are often realized by modifying a small range of content. If unlearning is performed on a large range, it may conflict with the original knowledge of the model, thus affecting the accuracy of the model in handling clean data. Therefore, we advocate selective unlearning to maintain the overall performance of the model. Second, determining the exact scope of the unlearning is a challenge. Intuitively, unlearning specific regions in the image (e.g., patches where triggers are located) seems to be a straightforward solution. However, given the diversity of attacks, especially attacks such as blended attacks, in which triggers are highly integrated with the normal recognized regions of the image, unlearning is exceptionally difficult. To address this challenge, we turn to unlearn discrete text tokens, a choice based on the observation that backdoors typically do not significantly overfit the semantic content of the text. We utilized ~\cite{chefer2021_ICCV_Mask} attribution methods to calculate the contribution values of each token in the text towards CLIP model predictions for image-text pairs.Subsequently, we retained tokens with higher contribution values, which are deemed likely to contain information relevant to the backdoor. In the following text, we represent this as $M_\theta(\cdot)$.

Furthermore, due to the high degree of similarity between the images and captions in the backdoor sample set. To enhance the effectiveness of the unlearning process, we adopt an innovative approach of performing Cartesian product combination on a subset of the few-shot unlearning as a way of data augmentation. This step generates a variety of combinations of backdoor samples with varying degrees of correlation, which significantly increases the data diversity and richness of the unlearning training. We refer to the above unlearning process as token-level local unlearning training, which can be formulated as follows:
% To enhance our model's resilience against backdoor attacks, we introduce a targeted forgetting strategy that mitigates the attacks' impact without compromising model accuracy. This strategy focuses on selective, not wholesale, forgetting and preserving model knowledge while addressing the minimal, yet crucial modifications introduced by backdoors. Given the complexity of identifying specific regions for forgetting, especially with sophisticated attacks that seamlessly blend triggers, we opt for discrete text token forgetting. This approach, informed by the observation that backdoors less frequently distort text semantics, involves evaluating each token's contribution to backdoor effects, as outlined by ~\cite{Chefer_2021_ICCV_Mask}, and selectively forgetting less impactful ones. 
% To further boost this process's efficiency, we employ data augmentation through Cartesian product combinations, enriching training data diversity. This method, known as token-based local forgetting, strategically strengthens the model against backdoor vulnerabilities.
% \setlength{\abovedisplayskip}{1em}

\begin{align}
    D_\text{mask} &= \Big\{ (\bm{I}_i, M_{\theta_\text{bd}}(\bm{T}_i)) \mid (\bm{I}_i, \bm{T}_i) \in (D_\text{topk} \times D_\text{topk}) \Big \},\\
    D_\text{unlearn} &= (D_\text{topk} \times D_\text{topk}) \cup D_{\text{mask}} , \\
    \label{eqa:unlearn}
    \theta_{\text{unlearn}} &= \min_{\theta}\{{\frac{1}{\lvert D_\text{unlearn} \rvert}\sum^{\lvert D_\text{unlearn} \rvert}_{i=1}{S_\theta(\bm{I}_i,\bm{T}_i)}}\},
\end{align}
where $D_\text{unlearn}$ is derived by extending $D_\text{topk}$ base on the above two conclusions . This process not only helps the model to identify and forget potential backdoor samples more effectively but also ensures that the ability to recognize normal samples is retained as much as possible while cutting down the backdoor influence. 

\begin{algorithm}
\caption{UBT Algorithm (Unlearning Backdoor Threats)}
\label{alg:UBT}
\renewcommand{\algorithmicrequire}{\textbf{Input:}}
\renewcommand{\algorithmicensure}{\textbf{Output:}}
\begin{algorithmic}[1]
\REQUIRE Training dataset: $D=\{\bm{I}_i,\bm{T}_i\}^N_{i=1}$, pretrained model: $\mathcal{M}$, size of $D_\text{susp}$: $S_\text{susp}$, size of $D_\text{topk}$: $S_\text{topk}$, 
\ENSURE  $\mathcal{M}_\text{unlearn}$   %%output
% \STATE $S = |D|$
\STATE Fine-tune $\mathcal{M}$ with the training dataset using Equation \eqref{eqa:FT} to obtain the poisoned model $\mathcal{M}_\text{bd}$
\STATE Calculate the similarity of $D$ using $\mathcal{M}$; select the smallest $S_\text{susp}$ as $D_\text{susp}$, and the remaining as $D_\text{normal}$
\STATE Fine-tune $\mathcal{M}_\text{bd}$ using Equation \eqref{eqa:overfitting} to obtain the overfitting model $\mathcal{M}_\text{overfitting}$
\STATE Calculate the similarity of $D_\text{susp}$ using $\mathcal{M}_\text{overfitting}$; select the largest $S_\text{topk}$ as $D_\text{topk}$
\STATE $D_\text{unlearn}= D_\text{topk} \times D_\text{topk} $
% \IF{Toxicity check is False}
%     \RETURN
% \ELSE
    \STATE Fine-tune $\mathcal{M}_\text{bd}$ using $D_\text{unlearn}$ and Equation \eqref{eqa:unlearn} to obtain the model $\mathcal{M}_\text{unlearn}$.
% \ENDIF
% \RETURN $\mathcal{M}_\text{unlearn}$
\end{algorithmic}
\end{algorithm}
\subsection{Analysis of the Existence of a Relatively Small Unlearning Dataset}
In this section, we analyze the upper bound on the minimum number of samples required for backdoor unlearning in poisoned models. We argue that, due to the small difference between clean and poisoned models, the number of samples needed for training should ideally be small, which provides insight into selecting a smaller unlearning set.
We use PAC-Bayes theory~\cite{mcallester1999pac} to demonstrate that a smaller unlearning dataset can effectively achieve the desired unlearning outcome. We first provide the definition of PAC-Bayes theory: Let the sample space be defined as $\mathcal{Z} = \mathcal{X} \times \mathcal{Y}$, where $\mathcal{X} = \mathcal{Y} = \mathbb{R}^n$. Let $D = \{x_i, y_i\}_{i=1}^{N}$ represent a dataset consisting of $N$ samples randomly drawn from the sample space, following a random probability distribution $P \in \mathcal{P}(\mathcal{Z})$. $\mathcal{P}(\mathcal{Z})$ denotes the family of probability measures over a set $\mathcal{Z}$ .Let $Q_0$ be a probability distribution over the hypothesis space $\mathcal{H}$, and after observing data $D$, the output probability distribution is denoted as $Q_D$. $l:\mathcal{Z}\times\mathcal{H} \rightarrow \mathbb{R}_+$ is the loss function. The PAC theory can be expressed as Theorem \ref{the:Bayese1}.
\begin{theorem}
\label{the:Bayese1}
For any $\delta \in (0,1)$, with probability at least $1 - \delta$, the following inequality holds:
\begin{equation}
\begin{aligned}
    \label{eqa:Bayese1}
    \mathbb{E}_{h\sim Q_D}&[L(h)]-\mathbb{E}_{h\sim Q_D}[\hat{L}(D,h)] \\
    &\leq \sqrt{\frac{1}{2n-1}(KL(Q_D||Q_0)+\log{\frac{n+2}{\delta}})}.
\end{aligned}
\end{equation}

For $h \in \mathcal{H}$, $L(h)=\mathbb{E}_{z \sim P}[l(z,h)]$ is the generalization risk, or simply risk, and $\hat{L}(D,h)=\frac{1}{|D|}\sum_{i=1}^{|D|}l((x_i,y_i),h)$ is the empirical loss. $KL(\cdot||\cdot)$ denotes the KL divergence.
\end{theorem}

Next, we will roughly analyze the impact of sample size on the distribution of model parameters before and after training. We will transform equation \ref{eqa:Bayese1} into:
\begin{equation}
\begin{aligned}
    \label{eqa:Bayese2}
    \mathbb{E}_{h\sim Q_D}&[L(h)]-\mathbb{E}_{h\sim Q_D}[\hat{L}(D,h)]\\
    &\leq \sqrt{\frac{1}{2n-1}KL(Q_D||Q_0)+\frac{\log{(n+2)}+C}{2n-1}}.
\end{aligned}
\end{equation}
 When \(\delta\) is fixed, \(C\) is a constant. 

\begin{lemma}
For any \(N\), there exists \(0 < \epsilon < 1\) such that when \(n > N\), the following inequality holds:
\begin{equation}
\frac{\log(n + 2)}{2n - 1} \leq \frac{1}{(2n - 1)^\epsilon}.
\label{lemma:epsilon_inequality}
\end{equation}
\end{lemma}

\begin{proof}
We start by analyzing the term \(\frac{\log(n + 2)}{2n - 1}\). Since \(\log(n + 2)\) grows logarithmically and \(2n - 1\) grows linearly with \(n\), for sufficiently large \(n\), the term 
\begin{equation}
\frac{\log(n + 2)}{2n - 1}
\end{equation}
will decay faster than any power of 
\begin{equation}
\frac{1}{(2n - 1)^\epsilon},
\end{equation}
where \(0 < \epsilon < 1\). Thus, there exists some \(N > 0\) such that for all \(n > N\), the inequality \eqref{lemma:epsilon_inequality} holds.
\end{proof}

\begin{lemma}
Under the condition \(n > N\), the following inequality holds:
\begin{align}
    \mathbb{E}_{h \sim Q_D}&[L(h)] - \mathbb{E}_{h \sim Q_D}[\hat{L}(D, h)] \nonumber \\
    &\leq \sqrt{\frac{1}{2n - 1} \left(KL(Q_D \| Q_0) + C\right) + \frac{1}{(2n - 1)^\epsilon}} \label{eqa:Bayese3} \\
    &\leq \sqrt{\frac{1}{(2n - 1)^\epsilon} \left(KL(Q_D \| Q_0) + C_0\right)}. \label{eqa:Bayese4}
\end{align}
\end{lemma}

\begin{proof}
Starting from Equation \eqref{eqa:Bayese3}, we apply the result from Lemma \ref{lemma:epsilon_inequality}. Substituting the term \(\frac{\log(n + 2)}{2n - 1} \leq \frac{1}{(2n - 1)^\epsilon}\), we obtain:

\begin{equation}
\begin{aligned}
\mathbb{E}_{h \sim Q_D}&[L(h)] - \mathbb{E}_{h \sim Q_D}[\hat{L}(D, h)] 
\\ &\leq \sqrt{\frac{1}{2n - 1} \left(KL(Q_D \| Q_0) + C\right) + \frac{1}{(2n - 1)^\epsilon}}.
\end{aligned}
\end{equation}

Further simplification gives:
\begin{equation}
\begin{aligned}
\mathbb{E}_{h \sim Q_D}&[L(h)] - \mathbb{E}_{h \sim Q_D}[\hat{L}(D, h)] 
\\&\leq \sqrt{\frac{1}{(2n - 1)^\epsilon} \left(KL(Q_D \| Q_0) + C_0\right)},
\end{aligned}
\end{equation}
which completes the proof.
\end{proof}

\begin{lemma}
For any \(r > 0\), a sufficient condition for 
\begin{equation}
\mathbb{E}_{h \sim Q_D}[L(h)] - \mathbb{E}_{h \sim Q_D}[\hat{L}(D, h)] \leq r
\end{equation}
is:
\begin{equation}
    \sqrt{\frac{1}{(2n - 1)^\epsilon} \left(KL(Q_D \| Q_0) + C_0\right)} \leq r.
\end{equation}
This implies:
\begin{equation}
n \geq \left(\frac{KL(Q_D \| Q_0) + C_0}{2r^2}\right)^\frac{1}{\epsilon} + \frac{1}{2} = N_0. \label{eqa:sample}
\end{equation}
\end{lemma}

\begin{proof}
Starting from the inequality:
\begin{equation}
\sqrt{\frac{1}{(2n - 1)^\epsilon} \left(KL(Q_D \| Q_0) + C_0\right)} \leq r,
\end{equation}
squaring both sides, we get:
\begin{equation}
\frac{1}{(2n - 1)^\epsilon} \left(KL(Q_D \| Q_0) + C_0\right) \leq r^2.
\end{equation}
This leads to the bound on \(n\):
\begin{equation}
n \geq \left(\frac{KL(Q_D \| Q_0) + C_0}{2r^2}\right)^\frac{1}{\epsilon} + \frac{1}{2} = N_0.
\end{equation}
Thus, \(n \geq N_0\) is the sufficient condition for 
\begin{equation}
\mathbb{E}_{h \sim Q_D}[L(h)] - \mathbb{E}_{h \sim Q_D}[\hat{L}(D, h)] \leq r.
\end{equation}
\end{proof}
When $r$ is sufficiently small, $N_0 \geq N$. Therefore, as long as Equation~\ref{eqa:sample} holds, we have
$\mathbb{E}_{h \sim Q_D}[L(h)] - \mathbb{E}_{h \sim Q_D}[\hat{L}(D, h)] \leq r$,
where $N_0$ is the estimated minimum number of samples required for training. Note that the estimation method used in Equation~\ref{eqa:Bayese4} introduces significant approximation errors, making the PAC upper bound not tight. As a result, $N_0$ does not accurately represent the minimum number of samples, and the actual minimum sample size $N_*$ satisfies $N_* \leq N_0$. Therefore, $N_0$ is an upper bound on the minimum sample size.

From Equation~\ref{eqa:sample}, we can see that $N_0$ is proportional to $KL(Q_D \| Q_0)$, which implies that the more similar the parameter distributions before and after training, the fewer samples are required for training. As shown in Figure~\ref{tab:dfs}, in the ``No defense`` scenario, the poisoned model and the retrained model are derived from training datasets with nearly identical quantities (differing by only about 0.3\% of backdoor samples). Consequently, the parameter distributions of the poisoned model$(Q_\text{bd})$ and the retrained model$(Q_\text{re})$ should be very similar (with a smaller KL divergence $KL(Q_\text{re} \| Q_\text{bd})$).
Although not explicitly shown in our paper, we can infer that the retrained model differs significantly from the pre-trained model$(Q_\text{pre})$ (with a larger KL divergence $KL(Q_\text{re} \| Q_\text{pre})$), as the pre-trained model lacks most of the knowledge in the training dataset, necessitating nearly the entire dataset for training. In fact, $KL(Q_\text{re} \| Q_\text{bd}) \approx \frac{1}{4} KL(Q_\text{re} \| Q_\text{pre})$, which assures us that fine-tuning does not require a large dataset to achieve unlearning. However, due to approximation errors, Equation~\ref{eqa:sample} cannot provide an accurate estimate of the dataset size, and through our experiments, we believe that using 1\% of the data is a good choice.
\section{Experiments}
\label{sec:experiments}
\subsection{Experimental Setting} 
We conduct backdoor attack experiments using a 500K subset of the CC3M dataset~\cite{sharma2018cc3m} and the CLIP model, with ViT/32-B as the visual encoder and Transformer as the text encoder. We add 1500 backdoor samples to this subset and utilize four backdoor attack methods: BadNet~\cite{gu2017badnets}, Blended~\cite{chen2017blended}, SIG~\cite{barni2019SIG}, SSBA~\cite{li2021SSBA}, and TrojVQA~\cite{walmer2022trojvqa}. The model is poisoned and trained with a batch size of 128 and a learning rate of 1e-6 for 3 epochs. We use ImageNet1K~\cite{deng2009imagenet} zero-shot classification task as the downstream task, selecting ``banana`` as the target label for the backdoor attack.

For backdoor defense, UBT first selects 1\% of the entire dataset(D) as suspicious data. We train an overfitting poisoned model with a batch size of 64 and a learning rate of 1e-6 for 5 epochs to make it challenging to generalize to clean data. Then, we further filter the dataset to include $\sqrt{|D| \cdot 1\%}$ of the data as unlearn data, where $| \cdot |$ denotes the size of the dataset. Although the MCL model has a higher poisoning rate compared to traditional models, we believe that the poisoning rate will not drop below a certain threshold (greater than $\sqrt{|D| \cdot 1\%}$) since attackers aim to maintain a high attack success rate. UBT uses unlearning techniques by adjusting the batch size to 64, the learning rate to 1e-5, and conducting 5 epochs of training to eliminate backdoor feature memories from the model, thereby enhancing security and robustness.
	
We use three methods for comparison:
\ding{182} ABL~\cite{li2021ABL}, as another method using data unlearning for backdoor defense. We employ the ABL method for CLIP as described in ~\cite{bansal2023cleanclip}, assuming $D_\text{susp} = D_\text{unlearn}$, and conduct unlearning defense. We use a batch size of 64 and a learning rate of 1e-6 for 10 epochs of training.  
\ding{183} RoCLIP~\cite{yang2024roclip} is considered the state-of-the-art defense method. We train it using a batch size of 128 and a learning rate of 1e-6, with the training epoch set to 24.
\ding{184} In the fine-tuning scenario, CleanCLIP~\cite{bansal2023cleanclip} is currently the state-of-the-art defense algorithm. We follow its specific experimental setup as described in the paper.

\subsection{Defense Performance with Multi-attacks}

\begin{table}[h]
  \centering
  \caption{The performance(\%) of UBT against five attack methods.}
  \label{tab:atk}
  \begin{tabular}{cccc}
    \toprule
    Attack Method & Defense Method &CA & ASR \\
    \midrule
    \multirow{2}{*}{BadNet~\cite{gu2017badnets}} & No defense & 62.61 & 80.92\\
    &UBT & 61.51 & 0.00 \\
    \midrule
    \multirow{2}{*}{Blended~\cite{chen2017blended}} & No defense & 62.58 & 97.99\\
    &UBT & 60.60 & 0.08 \\
    \midrule
    \multirow{2}{*}{SIG~\cite{barni2019SIG}} & No defense & 62.77 & 90.90\\
    &UBT & 62.70 & 0.27 \\
    \midrule
    \multirow{2}{*}{SSBA~\cite{li2021SSBA}} & No defense & 62.77 & 66.22\\
    &UBT & 62.20 & 4.33 \\
    \midrule
    \multirow{2}{*}{TrojVQA~\cite{walmer2022trojvqa}} & No defense & 62.45 & 96.19\\
    &UBT & 62.13 & 0.00 \\
  \bottomrule
\end{tabular}
\end{table}
In this part, we test the defense effectiveness of UBT under multiple attack methods. As shown in Table \ref{tab:atk}, we can draw the following conclusions: \ding{182} The UBT method demonstrates significant defense efficacy in various backdoor attack scenarios. It effectively reduces the ASR to close to or completely zero. \ding{183} The UBT method does not significantly impact model performance, maintaining high CA even with a substantial reduction in ASR (reducing by less than 2\% among the five methods). \ding{184} UBT's defense effectiveness on SSBA is relatively lower compared to methods like SIG, possibly because SSBA's backdoor triggers on images are more concealed.

\subsection{Comparing with SoTA Defense}
\noindent\textbf{Anti-backdoor unlearning} In this section, we compare the effectiveness of UBT and the backdoor defense method ABL. We design three experiments as follows: (1) ABL, (2) ABL with token-level unlearning algorithm, and (3) our backdoor defense method UBT. Additionally, we analyze the separation of clean samples and backdoor samples under these two strategies, as shown in Figure \ref{fig:ABL_vs_UBT}. The conclusions drawn from Table \ref{tab:ABL} are as follows:
\ding{182} The ABL method significantly reduces CA (by around 10\%), mainly due to the presence of a large number of clean samples in $D_\text{susp}$, leading to moda degradation of the performance of the modelring training through gradient ascent.
\ding{183} ABL increases ASR (BadNet from 80.92\% to 99.95\%, Blended from 97.99\% to 99.95\%). This is likely because the mixture of clean and backdoor samples in $D_\text{susp}$ prevents the model from finding backdoor features during unlearning, while the remaining backdoor samples in $D_\text{normal}$ strengthen the backdoor shortcut through contrastive loss during training.
\ding{184} Applying token-level unlearning strategy to ABL does not improve defense effectiveness (similar to the original results of ABL). This could be because token-level unlearning does not address the issue of mixed backdoor and clean samples in $D_\text{susp}$.
\ding{185} UBT effectively defends against backdoor attacks by successfully separating backdoor samples from clean samples and allowing the model to focus on unlearning backdoor features.
\begin{figure*}[t]
    \includegraphics[width=1.0\textwidth]{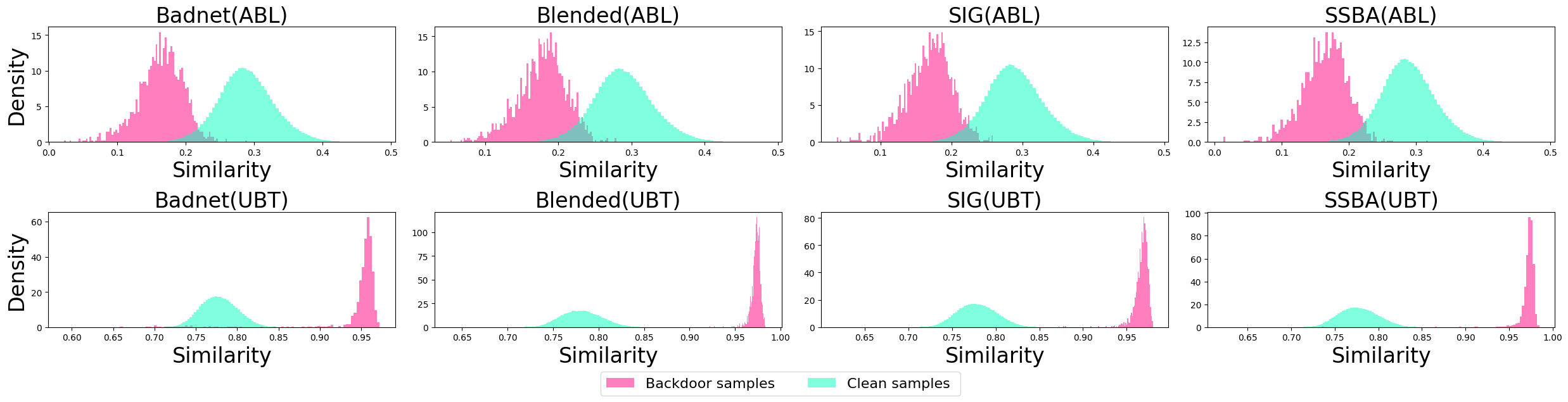}
    \caption{Comparison of the separation between backdoor samples (red) and clean samples (green) by UBT (top) and ABL (bottom) under 4 attack methods. The x-axis represents similarity, ranging from -1 to 1, and the y-axis represents density, indicating the proportion of all backdoor (clean) samples.}
    \label{fig:ABL_vs_UBT}
\end{figure*}
\begin{table}[h]
\centering
  \caption{The performance(\%)  of UBT and ABL against BadNet and Blended attacks.}
  \label{tab:ABL}
  \begin{tabular}{ccccc}
    \toprule
    \multirow{2}{*}{Method}&\multicolumn{2}{c}{BadNet~\cite{gu2017badnets}}&\multicolumn{2}{c}{Blended~\cite{chen2017blended}}\\
    &CA&ASR&CA&ASR \\
    \midrule
    No defense & 62.61 & 80.92 & 62.58 & 97.99 \\
    ABL~\cite{li2021ABL}& 51.55 & 89.63 &50.67&99.95\\
    ABL+Text Mask & 51.57 & 89.56 & 50.69  & 99.94\\
    UBT(ours) & 61.51 & \textbf{0.00} & 60.60 & \textbf{0.08}\\
  \bottomrule
\end{tabular}
\end{table}

\noindent\textbf{CleanCLIP and RoCLIP} We compared UBT with two state-of-the-art multimodal backdoor defense methods (CleanCLIP~\cite{bansal2023cleanclip} and RoCLIP~\cite{yang2024roclip}). We introduced the KL divergence from the retrained model as one of the metrics. This is typically used to compare the differences between a model trained with a unlearning algorithm and a completely retrained model, thereby evaluating the effectiveness of the unlearning algorithm. Here, we assess the effectiveness of the backdoor defense by comparing the differences between a model trained on clean data and the backdoor model after backdoor defense. Table \ref{tab:dfs} presents our experimental results. Compared to CleanCLIP and RoCLIP, our method exhibits greater advantages over CleanCLIP and RoCLIP by significantly reducing ASR (19\% decrease vs. CleanCLIP and 52\% decrease vs. RoCLIP), while maintaining superior CA (2.57\% increase vs. CleanCLIP and 1.05\% increase vs. RoCLIP). Additionally, we observed that CleanCLIP and RoCLIP are prone to causing model performance degradation, possibly due to CleanCLIP's use of an additional dataset for fine-tuning, leading to distribution discrepancies with the training data resulting in CA reduction. Furthermore, RoCLIP experiences performance declines at high poisoning rates, with the fine-tuning phase exhibiting a relatively higher poisoning rate (0.3\%).
\begin{table}[h]
\centering
  \caption{
Comparison of UBT and other defense methods against four attack methods. Our defense method achieved the best results (\%) under each attack. KL refers to the KL divergence between the retrain model.}
  \label{tab:dfs}
  \begin{tabular}{ccccc}
    \toprule
    Attack Method & Defense Method &CA & ASR & KL\\
    \midrule
    \multirow{4}{*}{BadNet~\cite{gu2017badnets}} & No defense & 62.61 & 80.92 & 0.035\\
    &CleanCLIP~\cite{bansal2023cleanclip} & 58.95 & 14.6 & 0.263\\
    &RoCLIP~\cite{yang2024roclip} & 61.04 & 63.41 & 0.145\\
    &UBT & 61.51 & \textbf{0.00} & \textbf{0.074}\\
    \midrule
    \multirow{4}{*}{Blended~\cite{chen2017blended}} & No defense & 62.58 & 97.99 & 0.034\\
    &CleanCLIP & 59.43 & 2.24 & 0.150\\
    &RoCLIP & 60.36 & 31.09 & 0.142\\
    &UBT & 60.60 & \textbf{0.08} & \textbf{0.072}\\
    \midrule
    \multirow{4}{*}{SIG~\cite{barni2019SIG}} & No defense & 62.77 & 90.90 &0.035\\
    &CleanCLIP & 59.44 & 48.48 &0.787\\
    &RoCLIP & 60.82 & 80.20 & 0.143\\
    &UBT & 62.70 & \textbf{0.27} & \textbf{0.039}\\
    \midrule
    \multirow{4}{*}{SSBA~\cite{li2021SSBA}} & No defense & 62.77 & 66.22 & 0.036\\
    &CleanCLIP & 58.90 & 15.53 & 0.199\\
    &RoCLIP & 60.61&40.05 &0.143\\
    &UBT & 62.20 & \textbf{4.33} & \textbf{0.044}\\
  \bottomrule
\end{tabular}
\end{table}
\subsection{Ablations}
\subsubsection{Overfitting Stage Loss Strategy}
In the overfitting model training phase, we partition the entire dataset into $D_\text{susp}$ and $D_\text{normal}$ and devise different loss strategies to enhance the model's sensitivity to backdoor samples compared to clean samples. To further validate our loss design, we utilize a 500K subset of the CC3M dataset and train with 1000 added backdoor samples using the BadNet~\cite{gu2017badnets} backdoor attack method. Then, we compute the cosine similarity of pairs of samples in $D_\text{susp}$ using the overfitting model. We visualize the impact of data filtering with $D_\text{normal}$ in Figure \ref{fig:5kvs5000k}. When training is solely based on $D_\text{susp}$, the similarity of both backdoor and clean samples relatively increases, but their distributions remain very close, making it challenging to distinguish between the two types of data. This is because during training, the model treats all samples equally in the overfitting phase. In contrast, incorporating $D_\text{normal}$ results in a more distinct difference between backdoor and clean samples. The mean similarity distribution of backdoor samples is around 0.94, while that of clean samples is around 0.76, due to the nature of contrastive learning, which treats backdoor and clean samples as negative samples, thus widening the gap between them. This statistical difference allows us to easily select a subset of backdoor samples for subsequent unlearning.
\begin{figure}[h]
    \includegraphics[width=0.45\textwidth]{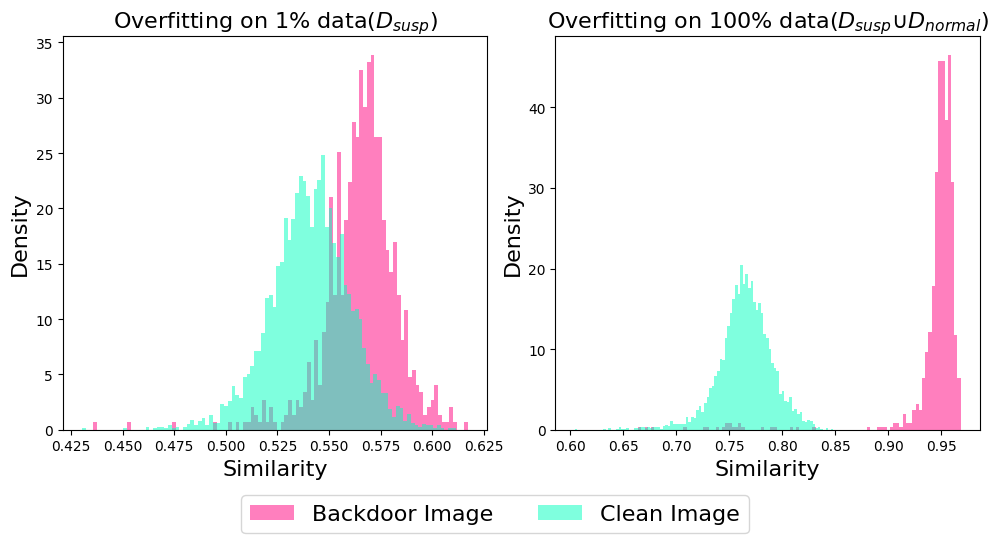}
    \caption{
Ablation studies on overfitting strategies: The left figure shows the results of overfitting using only $D_\text{susp}$, while the right figure shows the results of overfitting using the entire dataset $D$. Other settings of the images are the same as in Figure \ref{fig:ABL_vs_UBT}.}
    \label{fig:5kvs5000k}
     % \vspace{-5mm}
\end{figure}
\begin{figure*}[t]
    \centering
    \includegraphics[width=1.0\textwidth]{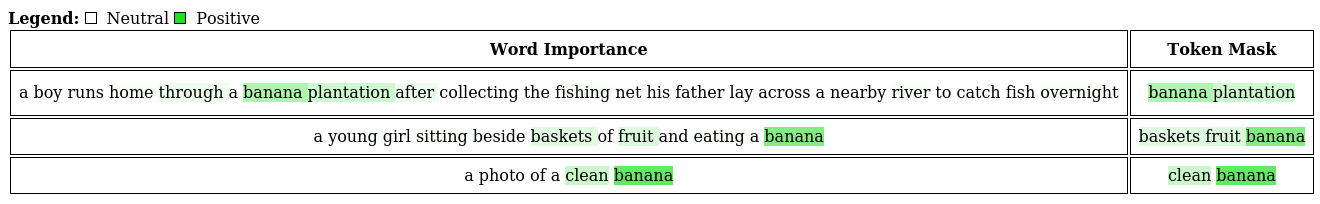}
    \caption{We use the attribution method from \cite{chefer2021_ICCV_Mask} to score the importance of each token, where green represents score. The darker the color, the higher the score. We choose a threshold of 0.1 and keep tokens with scores higher than this threshold.}
    \label{fig:text attribution}
    % \Description{}
\end{figure*}

\subsubsection{Token-Level Unlearning for Improved Unlearning}
In this section, we discuss the impact of token-level unlearning during the unlearning phase on defense performance improvement. Specifically, we adopt different unlearning strategies: (1) only using gradient ascent, (2) employing token-level unlearning on the visual modality, and (3) employing atoken-level unlearning on both the visual and textual modalities (4) UBT, which applies token-level unlearning on the text modality. As shown in Table \ref{tab:VM}, we can draw the following 
\begin{table}[t]
\centering
  \caption{
Ablation studies on the unlearning strategy of UBT. Results (\%) show that token-level unlearning on the text modality has the best performance.}
  \label{tab:VM}
  \begin{tabular}{cccc}
    \toprule
    Attack Method & Defense Method &CA & ASR \\
    \midrule
    \multirow{4}{*}{BadNet~\cite{gu2017badnets}} 
    & No defense & 62.61 & 80.92\\
     & GA & 61.29 & 0.01 \\
    & Image Mask + GA& 61.20 & 0.15\\
    &Image Mask + Text Mask +GA & 61.02 & 0.02 \\
    &UBT (Text Mask + GA ) & 61.51 & \textbf{0.00} \\
    \midrule
    \multirow{4}{*}{Blended~\cite{chen2017blended}} 
    & No defense & 62.58 & 97.99\\
    & GA & 60.81 & 0.16 \\
    & Image Mask + GA & 59.90 & 0.20 \\
    &Image Mask + Text Mask +GA & 60.33 & 0.15 \\
    &UBT (Text Mask + GA ) & 60.56 & \textbf{0.08} \\
  \bottomrule
  \vspace{-6mm}
\end{tabular}
\end{table}conclusions:\ding{182} Even with just using the gradient ascent strategy for unlearning, we achieve good defense results (ASR reduces to 0.01\% for BadNet and 0.16\% for Blended). This is because we accurately select a large number of backdoor samples, enabling precise unlearning of backdoor features during the unlearning phase.
\ding{183} Applying token-level unlearning to the visual modality does not improve the effectiveness of unlearning. We demonstrate the effect of our method on the image modality in Figure \ref{fig:image mask analysis}. Our research finds that for patch-level backdoor attacks like BadNet~\cite{gu2017badnets} and Trojvqa~\cite{walmer2022trojvqa}, we can identify hidden backdoor triggers in images. However, its defense effect is not as effective as GA (ASR increases by 0.14\%), possibly due to limitations in attribution algorithms, leading to bias in identifying specific images. For invisible attacks like Blended and SIG, it is very challenging to find trigger information through attribution methods. This is mainly because these attacks use global triggers that cover most or all areas of the image, while attribution methods focus more on local image details. Additionally, these triggers are integrated into the image, so even if we outline the trigger, we cannot remove other image information attached to the trigger (such as faces), thereby reducing the quality of unlearning.
\ding{184} UBT applies token-level unlearning to the text modality and achieves better defense results (ASR reduces to 0\% for BadNet and 0.08\% for Blended). This is because in the text modality, backdoor information is separated from other irrelevant information, allowing attribution methods to more accurately identify backdoor information.

\begin{figure*}[h]
    \includegraphics[width=1.0\textwidth]{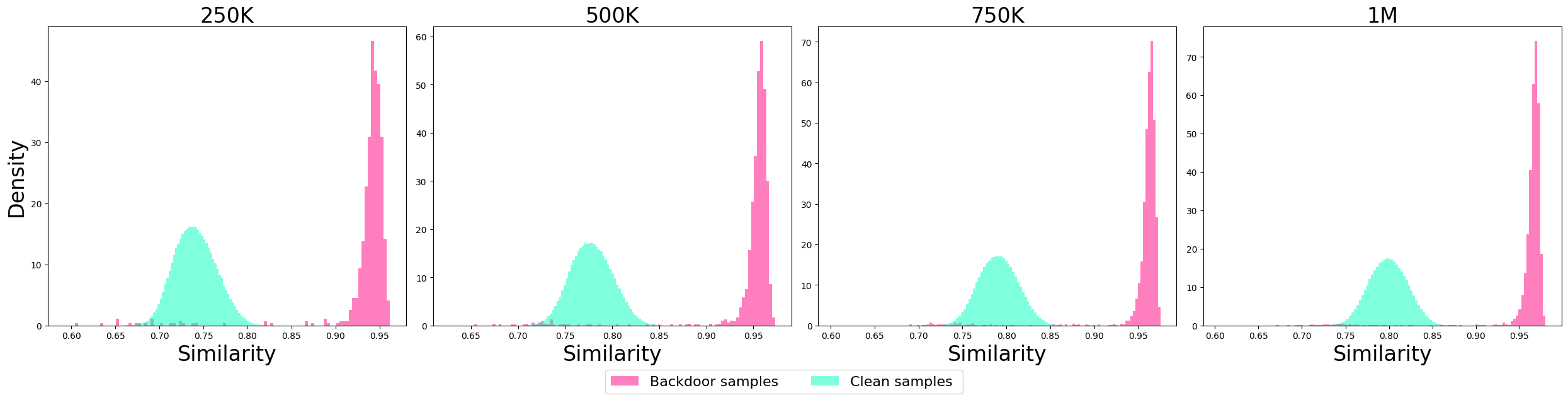}
    \caption{
The UBT method's filtering performance at a 0.3\% backdoor rate across datasets of different sizes. The numbers above the images represent the dataset sizes. Other settings are the same as in Figure~\ref{fig:ABL_vs_UBT}.
}
    \label{fig:data_size}
    % \Description{}
    \vspace{-4mm}
\end{figure*}

\begin{figure*}[h]
    \includegraphics[width=1.0\textwidth]{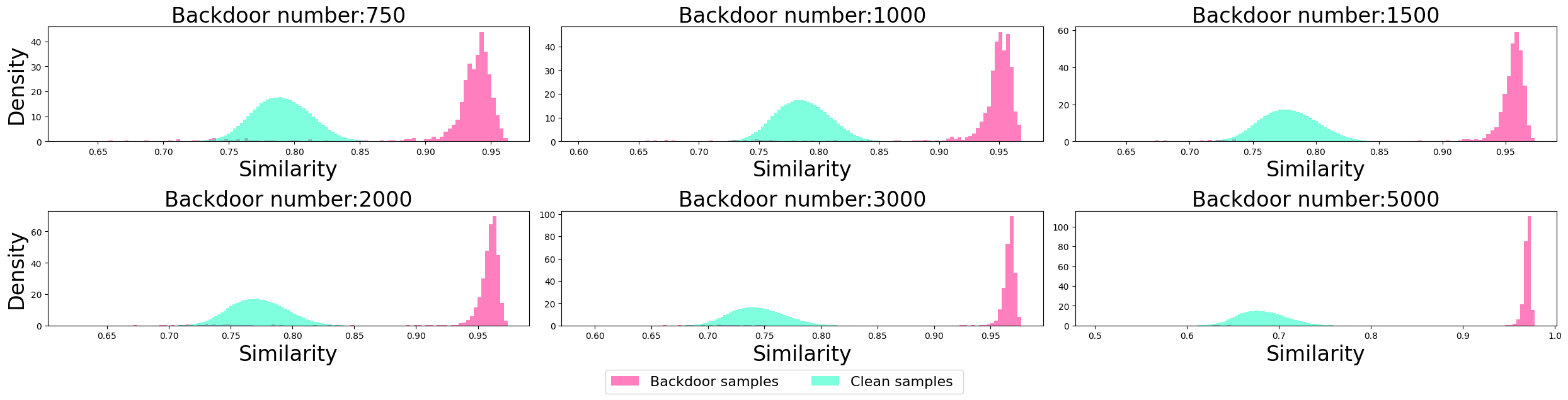}
    \caption{
The UBT method's performance at different backdoor rates using a dataset size of 500K. The numbers above each image represent the backdoor number. Other settings are the same as in Figure~\ref{fig:ABL_vs_UBT}.
}
    \label{fig:bd_num}
    % \Description{}
    \vspace{-4mm}
\end{figure*}

\subsection{The Damage to Clean Models}
Machine unlearning has a significant impact on models, requiring careful consideration. It's crucial to assess model poisoning before unlearning to avoid unnecessary actions on clean models. In fact, we can judge the poisoning of the model based on the distribution of similarities in $D_\text{unlearn}$, where a poisoned model's $D_\text{unlearn}$ should have a more concentrated and higher similarity. As shown in Table \ref{tab:CleanM}, we can draw the following conclusions: \ding{182} CleanCLIP does not consider the scenario of clean models. If the defender mistakenly uses clean samples for unlearning, CleanCLIP will reduce the model's CA. \ding{183} UBT judges whether the model is poisoned before unlearning, thus rejecting the unlearning of clean models to avoid unnecessary performance loss. \textit{Other experiments are detailed in the supplementary material.}

\begin{table}[h]
\centering
  \caption{Different defense methods applied to clean models and their impact on CA.(\%) }
  \label{tab:CleanM}
  \begin{tabular}{cccc}
    \toprule
     & No defense & CleanCLIP~\cite{bansal2023cleanclip} & UBT \\
    \midrule
    CA & 62.69 & 59.38 & 62.69 \\
  \bottomrule
\end{tabular}
\end{table}
\subsection{Performance of UBT on Different Downstream Datasets}
\begin{table*}[h]
\centering
  \caption{
Comparison of UBT, CleanCLIP, and other defense methods on 6 downstream datasets. Our defense method achieved the best results (\%) on each dataset.}
  \label{tab:dataset}
  \begin{tabular}{cccccccccccccc}
    \toprule
    \multirow{2}{*}{Defense Method} & \multicolumn{2}{c}{CIFAR10} & \multicolumn{2}{c}{CIFAR100} & \multicolumn{2}{c}{Caltech101} & \multicolumn{2}{c}{DTD} & \multicolumn{2}{c}{OxfordIIITPet} & \multicolumn{2}{c}{Food101} \\
    \cmidrule(lr){2-3} \cmidrule(lr){4-5} \cmidrule(lr){6-7} \cmidrule(lr){8-9} \cmidrule(lr){10-11} \cmidrule(lr){12-13}
    & CA & ASR & CA & ASR & CA & ASR & CA & ASR & CA & ASR & CA & ASR \\
    \midrule
    No defense & 90.27 & 99.87 & 65.5 & 99.74 & 78.3 & 68.58 & 43.03 & 77.77 & 79.34 & 63.58 & 78.98 & 41.25 \\
    CleanCLIP & 85.74 & 50.27 & 62.15 & 91.01 & 77.10 & 8.40 & 39.84 & 24.84 & 78.50 & 12.68 & 75.06 & 11.21 \\
    UBT & 87.99 & \textbf{0.01} & 65.05 & \textbf{0.00} & 77.97 & \textbf{0.02} & 43.19 & \textbf{0.00} & 79.20 & \textbf{0.00} & 79.67 & \textbf{0.01} \\
  \bottomrule
\end{tabular}
\end{table*}
In this subsection, we compare the defensive effectiveness of UBT and CleanCLIP on different downstream datasets. 
\begin{figure}[h]
    \includegraphics[width=0.45\textwidth]{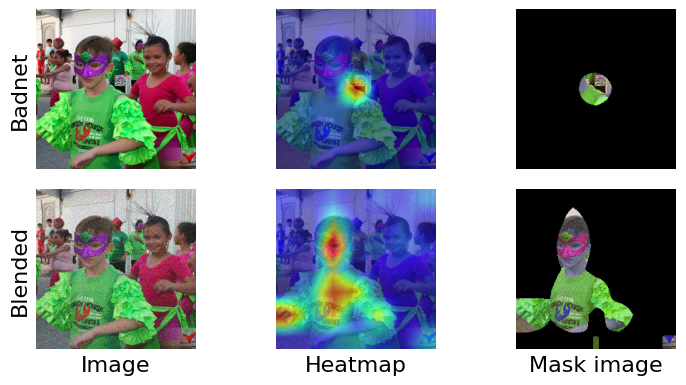}
    \caption{
We utilize the attribution method by~\cite{chefer2021_ICCV_Mask} to score token importance and display it via a heatmap, keeping tokens with scores above 0.3. We then examine BadNet (top) and Blended (bottom).
}
    \label{fig:image mask analysis}
    % \Description{}
    \vspace{-4mm}
\end{figure}
We use BadNet~\cite{gu2017badnets} as the backdoor attack method and designed different attack target labels for various datasets. The results are shown in Table~\ref{tab:dataset}, and we can draw the following conclusions: \ding{182} The defensive effectiveness of CleanCLIP is inconsistent across datasets. For example, it reduces the ASR by 69\% on Caltech101~\cite{FeiFei2004LearningGV} but only by 8\% on CIFAR100. In contrast, UBT successfully eliminates the backdoor threat on every dataset (ASR $\leq$ 1\%). This is because CleanCLIP attempts to mitigate the backdoor influence by fine-tuning with a large number of clean samples, without possessing knowledge of the backdoor samples. UBT, on the other hand, uses an overfitting model to identify and forget backdoor-related knowledge. \ding{183} CleanCLIP significantly reduces the accuracy of the model across different datasets, while UBT maintains accuracy. This is because the image-text pairs constructed by CleanCLIP have a distribution mismatch with the model’s training data, leading to degraded model performance. UBT uses token-level unlearning techniques, allowing the model to focus on unlearning backdoor triggers without affecting overall performance.

\subsection{Performance of UBT in Image-Text Retrieval Tasks}
In this subsection, we explore the defensive effectiveness of UBT in the image-text retrieval downstream task. We use BadNet~\cite{gu2017badnets} as the backdoor attack method in our experiments, with ``a photo of banana`` as the target label. Experiments were conducted on the Flickr30K~\cite{flickr30k} and COCO~\cite{chen2015microsoft} datasets. The results are shown in Table~\ref{tab:Retrival}, and we can draw the following conclusions: UBT can successfully defend against backdoors even in more complex image-text retrieval tasks. This success is attributed to the unlearning algorithm, which enables the model to forget backdoor-related knowledge, thereby ensuring strong performance across any task.

\begin{table}[h]
\centering
  \caption{
Comparison of UBT and other defense methods on the retrieval task. Our defense method achieved the best results (\%).}
  \label{tab:Retrival}
  \begin{tabular}{cccc}
    \toprule
    Dataset & Defense Method &CA & ASR \\
    \midrule
    \multirow{4}{*}{Flickr30K} & No defense & 77.50 & 80.80\\
    &CleanCLIP~\cite{bansal2023cleanclip} & 75.00 & 26.2 \\
    &RoCLIP~\cite{yang2024roclip} & 76.50& 33.7 \\
    &UBT & 76.90 & \textbf{0.1} \\
    \midrule
    \multirow{4}{*}{COCO} & No defense & 50.36 & 67.76\\
    &CleanCLIP & 46.98 & 14.92 \\
    &RoCLIP & 48.44& 9.08 \\
    &UBT & 49.36 & \textbf{0.06} \\
  \bottomrule
\end{tabular}
\end{table}
\subsection{Comparison of Training Time between UBT and Other Defense Methods}
In this section, we compare the training rates of UBT with previous defense methods, as shown in Table~\ref{tab:time}. We can draw the following conclusions. \ding{182} Our method is faster compared to other defense methods (reduced training time by 28\%) because we employ a carefully designed backdoor filtering strategy, achieving unlearn of the backdoor with fewer samples. \ding{183} RoCLIP~\cite{yang2024roclip} takes longer to train compared to other poisoning methods, possibly because RoCLIP needs to defend against backdoor attacks during training and introduces a defense strategy using text feature pools, which slows down convergence and prolongs training time.
\begin{table}[h]
    \centering
    \caption{Comparison of time usage between UBT and Other backdoor defense methods, where ``Training Time`` refers to the time taken for fine-tuning training, i.e., the training time of the poisoned model, and ``Defense Time`` refers to the time spent defending against fine-tuning of the poisoned model.(second)}
    \begin{tabular}{cccc}
    \toprule
    Defense method & Training time & Defense time & Total Time\\
    \midrule
    No Defense & 2826 & - & 2826 \\
    CleanCLIP\cite{bansal2023cleanclip} & 2826 & 11400 & 14226 \\
    RoCLIP\cite{yang2024roclip} & 53946 & - & 53946\\
    UBT & 2826 & 7402 & \textbf{10228} \\
    \bottomrule
    \end{tabular}
    \label{tab:time}
\end{table}
\subsection{The Defensive Effect of UBT at Different Data Scales}

In this section, we explore the impact of the number of backdoor samples on defensive effectiveness. Taking BadNet as an example, we analyze the influence of sample quantity on our method from two perspectives: (1) different numbers of backdoor samples; (2) different dataset sizes.

First, we fix the dataset size at 500K and vary the number of injected backdoor samples. The results are shown in Figure~\ref{tab:br}. We can draw the following conclusions: 
\ding{182} The UBT method performs well at high poisoning rates. This is because we expand the backdoor unlearning dataset through the Cartesian product, and a larger unlearning dataset ensures that the unlearning algorithm maintains good performance in handling high poisoning rate models. 
\ding{183} The UBT method also performs well at low poisoning rates. This is because we only need to filter out a very small number of samples from the dataset to achieve unlearning ($\sqrt{|D| \cdot 1\%}$). Our overfitting model training-based filtering method ensures that backdoor samples are separated at low poisoning rates, preventing clean samples from being mixed in, thus ensuring the effectiveness of unlearning.

Next, we fix the poisoning rate and vary the dataset size. The results are shown in Figure~\ref{tab:datasize}. We can conclude that UBT also achieves good defensive results when faced with different scales of training datasets. This is because UBT dynamically adjusts the scale of the suspicious dataset and the unlearning dataset based on the size of the dataset, enabling effective defense.

Additionally, we visualize the separation between backdoor samples and clean samples under the above two scenarios.As shown in Figure~\ref{fig:data_size}, UBT effectively separates backdoor data from clean data across different dataset sizes under the same poisoning rate. It can be observed that, regardless of the dataset size, UBT accurately identifies and separates backdoor samples from clean samples. Figure~\ref{fig:bd_num} further illustrates the effect under varying poisoning rates. Even at lower poisoning rates, UBT successfully isolates backdoor samples from the dataset. Moreover, as the poisoning rate increases, the UBT method amplifies the similarity gap between backdoor samples and clean samples, making them easier to distinguish.

\begin{table}[h]
\centering
  \caption{
Performance (\%) of UBT and CleanCLIP against BadNet attacks with varying numbers of backdoor samples. ``Backdoor Number`` refers to the count of backdoor samples in the training dataset.The training data size is fixed at 500K.}
  \label{tab:br}
  \begin{tabular}{cccc}
    \toprule
    Backdoor Number & Defense Method &CA & ASR \\
    \midrule
    \multirow{2}{*}{750} & No defense & 62.80 & 66.92\\
    % &CleanCLIP~\cite{bansal2023cleanclip} & 59.58 & 16.62 \\
    &UBT & 61.87 & \textbf{0.08} \\
    \midrule
    \multirow{2}{*}{1000} & No defense & 62.79 & 72.45\\
    % &CleanCLIP & 59.00 & 19.30 \\
    &UBT & 61.81 & \textbf{0.02} \\
    \midrule
    \multirow{2}{*}{1500} & No defense & 62.61 & 80.92\\
    % &CleanCLIP & 58.95 & 14.6 \\
    &UBT & 61.51 & \textbf{0.00} \\
    \midrule
    \multirow{2}{*}{2000} & No defense & 62.72 & 81.00\\
    % &CleanCLIP & 58.80 & 29.48\\
    &UBT & 61.95 & \textbf{0.06} \\
    \midrule
    \multirow{2}{*}{3000} & No defense & 62.29 & 84.69\\
    % &CleanCLIP & 58.90 & 15.53 \\
    &UBT & 62.08 & \textbf{0.00} \\
    \midrule 
    \multirow{2}{*}{5000} & No defense & 62.74 & 89.20\\
    % &CleanCLIP & 59.31 & 34.21 \\
    &UBT & 61.77 & \textbf{0.00} \\
  \bottomrule
\end{tabular}
\end{table}

\begin{table}[h]
\centering
  \caption{
Performance (\%) of UBT and CleanCLIP against BadNet attacks with varying dataset sizes. ``Dataset Size`` refers to the number of training samples in the dataset.The backdoor number is fixed at 1500.}
  \label{tab:datasize}
  \begin{tabular}{cccc}
    \toprule
    Datasets size & Defense Method &CA & ASR \\
    \midrule
    \multirow{2}{*}{250K} & No defense & 62.69 & 65.80\\
    % &CleanCLIP~\cite{bansal2023cleanclip} & 59.80 & 5.02 \\
    &UBT & 62.29 & \textbf{0.32} \\
    \midrule
    \multirow{2}{*}{500K} & No defense & 62.61 & 80.92\\
    % &CleanCLIP & 58.95 & 14.6 \\
    &UBT & 61.51 & \textbf{0.00} \\
    \midrule
    \multirow{2}{*}{750K} & No defense & 62.98 & 80.02\\
    % &CleanCLIP & 58.74 & 12.85 \\
    &UBT & 62.17 & \textbf{0.04} \\
    \midrule
    \multirow{2}{*}{1M} & No defense & 62.73 & 86.79\\
    % &CleanCLIP & 58.86 & 25.22\\
    &UBT & 62.32 & \textbf{0.46} \\
  \bottomrule
\end{tabular}
\end{table}
\section{Future work} 
As MCL models become more prevalent across various applications, the threat of backdoor attacks has intensified. To tackle this challenge, ongoing research is focused on refining defense strategies. The following sections discuss advancements in precise backdoor localization, general backdoor defense methods, and low-cost, rapid defense solutions.

1. \textbf{More Precise Backdoor Localization Strategies}  
As backdoor attacks evolve, they are becoming increasingly sophisticated and covert. Future efforts must focus on developing more accurate and efficient filtering strategies to detect backdoors within poisoned models. However, existing methods, such as ABL~\cite{bansal2023cleanclip}, face limitations, especially with large datasets in MCL models. Improved localization strategies are essential to adapt to evolving attack patterns, thereby enhancing MCL model security.

2. \textbf{More General Backdoor Defense Methods}  
MCL models undergo training in pre-training and fine-tuning stages, where backdoors can be inserted at any phase. Current defenses, like CleanCLIP and RoCLIP, are effective only in specific stages. Future research should prioritize developing general defense strategies applicable throughout all training phases, boosting the overall defense against diverse attacks.

3. \textbf{Lower-Cost, Faster Backdoor Defense Methods}  
As MCL models scale up with larger datasets, existing defense methods may become time-intensive. Future strategies must aim to reduce defense costs and improve efficiency, ensuring that MCL models remain secure and effective during large-scale training.

Future research should focus on precise localization, general defense strategies, and cost reduction to strengthen MCL model security, ensuring their reliability in diverse real-world scenarios.

\section{Conclusion}
This study proposes UBT, a defense strategy against backdoor attacks in multimodal contrastive learning. UBT enhances the model's sensitivity to backdoor triggers by overfitting the poisoned model, thereby identifying a portion of backdoor samples from a large dataset. With only a few selected backdoor samples, it constructs poisoned pairs and employs token-level local unlearning to effectively break the backdoor shortcuts in the poisoned model. We experimentally validated the effectiveness of this method in reducing the success rate of attacks and maintaining the accuracy of model purification, offering a new defense approach for the security of multimodal contrastive learning.

% \newpage
% \bibliographystyle{IEEEtran}
% \bibliography{IEEEabrv,main}

\end{document}